\documentclass[a4paper,twocolumn,superscriptaddress,11pt,accepted=2018-12-13]{quantumarticle}
\pdfoutput=1

\usepackage[utf8]{inputenc}
\usepackage{textcomp}
\usepackage{bbm}

\usepackage[numbers,sort&compress]{natbib}

\usepackage{amsfonts}
\usepackage{amsmath}
\usepackage{amsthm}
\usepackage{braket}
\usepackage{hyperref}
\usepackage{epsfig,graphics,graphicx}
\usepackage{subfig}

\usepackage{xcolor}
\definecolor{quantumviolet}{HTML}{53257F}

\hypersetup{citecolor=quantumviolet}
\hypersetup{colorlinks=true}
\hypersetup{linkcolor=quantumviolet}
\hypersetup{urlcolor=quantumviolet}

\newcommand{\minev}{\lambda_\mathrm{min}}
\DeclareMathOperator{\diag}{diag}
\DeclareMathOperator{\tr}{tr}
\DeclareMathOperator{\rank}{rank}

\newtheorem{lemma}{Lemma}

\newcommand{\norm}[1]{{\lVert{#1}\rVert}}

\newcommand{\ketbrad}[1]{\left|{#1}\rangle\!\langle{#1}\right|}

\begin{document}
\title{Bound entangled states fit for robust experimental verification}
\author{Gael Sent\'is}
\affiliation{Naturwissenschaftlich-Technische Fakult\"at, Universit\"at Siegen, 57068 Siegen, Germany}
\affiliation{Departamento de F\'isica Te\'orica e Historia de la Ciencia, Universidad del Pa\'is Vasco UPV/EHU, E-48080 Bilbao, Spain}
\email{gael.sentis@uni-siegen.de}
\orcid{https://orcid.org/0000-0002-4982-6570}
\author{Johannes N. Greiner}
\affiliation{3rd Institute  of  Physics,  University  of  Stuttgart  and  Institute  for  Quantum  Science  and  Technology, IQST,  Pfaffenwaldring  57,  D-70569  Stuttgart,  Germany}
\author{Jiangwei Shang}
\affiliation{Beijing Key Laboratory of Nanophotonics and Ultrafine Optoelectronic Systems, School of Physics, Beijing Institute of Technology, Beijing 100081, China}
\affiliation{Naturwissenschaftlich-Technische Fakult\"at, Universit\"at Siegen, 57068 Siegen, Germany}
\email{jiangwei.shang@bit.edu.cn}
\orcid{https://orcid.org/0000-0002-2332-5882}
\author{Jens Siewert}
\affiliation{Departamento de Qu\'imica F\'isica, Universidad del Pa\'is Vasco UPV/EHU, E-48080 Bilbao, Spain}
\affiliation{IKERBASQUE Basque Foundation for Science, E-48013 Bilbao, Spain}
\email{jens.siewert@ehu.eus}
\orcid{https://orcid.org/0000-0002-9410-5043}
\author{Matthias Kleinmann}
\affiliation{Naturwissenschaftlich-Technische Fakult\"at, Universit\"at Siegen, 57068 Siegen, Germany}
\affiliation{Departamento de F\'isica Te\'orica e Historia de la Ciencia, Universidad del Pa\'is Vasco UPV/EHU, E-48080 Bilbao, Spain}
\email{matthias.kleinmann@uni-siegen.de}
\orcid{https://orcid.org/0000-0002-5782-804X}

\begin{abstract}
Preparing and certifying bound entangled states in the laboratory is an intrinsically hard task, due to both the fact that they typically form narrow regions in state space, and that a certificate requires a tomographic reconstruction of the density matrix. Indeed, the previous experiments that have reported the preparation of a bound entangled state relied on such tomographic reconstruction techniques. However, the reliability of these results crucially depends on the extra assumption of an unbiased reconstruction. We propose an alternative method for certifying the bound entangled character of a quantum state that leads to a rigorous claim within a desired statistical significance, while bypassing a full reconstruction of the state. The method is comprised by a search for bound entangled states that are robust for experimental verification, and a hypothesis test tailored for the detection of bound entanglement that is naturally equipped with a measure of statistical significance. We apply our method to families of states of $3\times 3$ and $4\times 4$ systems, and find that the experimental certification of bound entangled states is well within reach.
\end{abstract}

\maketitle

\section{Introduction}

To experimentally prepare, characterize and control entangled quantum states is an essential item in the development of quantum-enhanced technologies, but it also serves the indispensable purpose of testing the predictions of entanglement theory in the laboratory. Among the most intriguing features of this theory stands the existence of \emph{bound} entanglement~\cite{Horodecki1998}, a form of entanglement that cannot be distilled into
singlet states by any protocol that uses only local operations and classical communication. Originally considered as useless for quantum information processing, bound entangled states were later established as a valid resource in the contexts of quantum key distribution~\cite{Horodecki2005}, entanglement activation~\cite{Horodecki1999a,Masanes2006}, metrology~\cite{Czekaj2015,Toth2018}, steering~\cite{Moroder2014a}, and nonlocality~\cite{Vertesi2014}, and their nondistillability has been linked to irreversibility in thermodynamics~\cite{Horodecki2002,Brandao2008}. 

Complementing these theoretical achievements, substantial efforts have been 
devoted to experimentally producing and verifying bound entanglement. The first 
experimental report on the preparation of a bound entangled state was presented 
in~\cite{Amselem2009}, although the result was disputed~\cite{Lavoie2010a} and 
subsequently amended~\cite{Lavoie2010}. The state prepared was the four-qubit 
Smolin state~\cite{Smolin2001}, thus it showcases a multipartite instance of 
bound entanglement, which is fundamentally distinct from the bipartite case: 
when multiple parties are present, entanglement can still potentially be 
distilled if two of the parties work together. Further experimental works on 
multipartite bound entanglement include 
Refs.~\cite{Barreiro2010,Kampermann2010,Dobek2011,Kaneda2012,Amselem2013,Dobek2013}. 
Examples of experimental bipartite setups are more scarce. 
In Ref.~\cite{DiGuglielmo2011} bipartite bound entanglement was produced using 
four-mode continuous-variable Gaussian states, and Ref.~\cite{Hiesmayr2013a} 
focuses on a family of two-qutrit states.

Since the property of nondistillability is experimentally inaccessible in a direct manner, 
a natural route to verify a state as bound entangled is to do
a full tomographic reconstruction of the density matrix from the experimental data~\cite{Paris2004} and apply the only known computable criterion on it~\cite{Horodecki1998}: if an entangled state has positive partial transpose (PPT), then it is nondistillable and therefore bound entangled\footnote{It is still an open question whether the PPT criterion completely characterizes bound entanglement, namely whether all nondistillable states are PPT.}. However, it has recently been pointed out that widely used reconstruction methods like maximum likelihood and least squares~\cite{Hradil1997,James2001} inevitably suffer 
from bias~\cite{Sugiyama2012,Schwemmer2015}, caused by
imposing a positivity constraint over compatible density matrices.
In some cases, the bias can be large enough to drastically change the entanglement properties of the estimator with respect to the true state. 
In addition to this state of affairs, 
the variance of the estimator is
usually calculated by bootstrapping~\cite{Efron1994}, which only accounts for statistical fluctuations, and can result in a smaller variance than the actual bias of the estimator~\cite{Schwemmer2015}. 
In contrast, a direct reconstruction of the state by linear inversion produces an unbiased estimator, but at the cost of admitting unphysical density matrices. 
Then, 
the PPT criterion simply looses all meaning.

All the experimental works cited above support their claims on some combination of maximum likelihood or least squares reconstruction and bootstrapping.
There exist more informative methods to derive errors from tomographic data, 
such as credibility~\cite{Shang2013}
and confidence~\cite{Christandl2012,Blume-Kohout2012} regions, and also the alternative of using linear inversion in addition to a sufficiently large number of measurements that guarantees physical estimates~\cite{Knips2015}. 
Should these methods be applied to the detection of a bound entangled state, more robust results may be generated, although they might come at the expense of being computationally expensive or even intractable~\cite{Suess2016}.
However, regardless of the reconstruction method of choice, the problem of experimentally testing bound entanglement is intrinsically challenging. This is so because bound entangled regions of the state space are typically very small in volume~\cite{Zyczkowski1999}. Furthermore, at least for the known cases in low-dimensional systems, 
bound entangled states are close to both the sets of separable states and distillable entangled states. 
This translates into the requirement of a highly precise experimental setup, and 
the deepening of the potential pitfalls of 
biased tomographic reconstructions.

In this paper, we set ourselves to improve on the above situation by devising methods that enable robust experimental certification of bound entanglement.
Instead of advocating for a particular tomographic method for detecting bound entanglement or considering the preparation of a specific state, we address the  
more generic question: Which are the best candidate states for an experimental verification of bound entanglement? In other words, for 
bipartite systems, we aim at finding states that have the largest ball of bound entangled states around them~\cite{Bandyopadhyay2008}.
To this end, we construct simple lower bounds that the radius $r$ of such a ball (in Hilbert--Schmidt distance) has to obey for a given state, and formulate an optimization problem that maximizes $r$ over parametrized families of states. 
Having a value for the maximum radius, $r^*$, allows us to assert the robustness of the target state at the center of the ball, 
and in turn gives us an idea of the required number of preparations of the state in a potential experiment. We proceed by designing a $\chi^2$ hypothesis test directly applicable over unprocessed tomographic data that provides a certificate for bound entanglement within some statistical significance. 
We show that our proposed method, combining the search of a robust candidate state with a statistical analysis through a hypothesis test, makes rigorous bound entanglement verification not only experimentally feasible with current technology, but also computationally cheap.

The paper is structured as follows. First we derive the constraints for the existence of a ball of bound entangled states of radius $r$ around a generic bipartite state of dimension $d$. Then, we apply our method to two families of states of dimensions $3\times 3$ and $4\times 4$, known to contain bound entanglement, and find robust candidates for its experimental verification. 
We proceed to devise a hypothesis test for bound entanglement, and test the robustness of the selected candidate states in terms of the necessary number of samples to achieve a statistically significant certification under realistic experimental conditions.

\section{A bound entangled ball}

Verifying the bound entangled character of a bipartite state 
$\rho$ requires, on the one hand, showing that it is entangled, and on the other hand 
proving that the entanglement of $\rho$ is nondistillable. Non-distillability is usually verified via the (sufficient) PPT condition,
which we denote by \mbox{$\Gamma(\rho)\geq 0$}.
Throughout this paper, we identify bound entangled states with PPT entangled states.
As for verifying that $\rho$ is entangled, there exist several inequivalent criteria. We choose the violation of the \emph{computable cross norm} or \emph{realignment} criterion (CCNR)~\cite{Guhne2009}, since it is simple and is generally tight enough to detect bound entanglement. The CCNR criterion dictates that, for some local orthonormal basis [with respect to the Hilbert--Schmidt inner product $(A,B)= \tr(A^\dag B)$] of the Hermitian matrices $\{g_k\}$, $\rho$ is entangled if the matrix $R(\rho)_{k,l}=\tr(\rho\, g_k\otimes g_l)$ obeys $\norm{R(\rho)}_1> 1$, where $\norm{Y}_1=\tr\sqrt{Y^\dagger Y}$ is the trace norm of $Y$. The points in the state space that violate CCNR, fulfill PPT, and correspond to physical states (that is, satisfy the positivity condition $\rho\geq 0$), thus define a volume of bound entangled density matrices. 

Given a bound entangled state $\rho$ that obeys these conditions, we inquire how far we can move away from it while remaining in the bound entangled region. We construct a new state $\tau=\rho+rX$, where $r\geq 0$, 
and $X$ is a traceless Hermitian matrix with bounded Hilbert--Schmidt norm $\norm{X}_2=\sqrt{\tr X^\dagger X}\le 1$.
The set of all such matrices forms a Hilbert--Schmidt ball that we denote by $\mathcal{B}'$. We then define the set of all states $\tau$ as $\mathcal B(\rho,r):=\rho+r\mathcal B'$.
Note that $\norm{\rho-\tau}_2 \leq r$. We can bound the minimum eigenvalue of $\tau$ as
\begin{equation}\begin{split}\label{first_lower_bound}
 \minev(\tau)&\ge \minev(\rho)-r\norm{X}_\infty\\&
 \ge \minev(\rho)-r\sqrt{(d-1)/d}\,,
\end{split}\end{equation}
where $\norm{X}_\infty=\lambda_{\rm max}(X)$ is the uniform norm of $X$. The first inequality holds since, in the extreme case, the eigenvector associated to the minimal eigenvalue of $X$ is aligned with the corresponding one for $\rho$.
For the second inequality we have used that 
%
\begin{equation}\label{lemma1}
\max\set{\norm{X}_\infty| X\in \mathcal B'}=  \sqrt{(d-1)/d} 
\end{equation}
%
(we provide a proof of this equation in Appendix~\ref{app:proofs}).

Similarly, since the partial transpose $\Gamma(\tau)$ does not change the 
 Hilbert--Schmidt ball, $\Gamma(\mathcal B) = \mathcal B$, we have
\begin{equation}\label{ppt-bound}
 \minev[\Gamma(\tau)] \ge \minev[\Gamma(\rho)]-r\sqrt{(d-1)/d} \,.
\end{equation}

We can bound the value of the CCNR criterion over $\tau$ in a similar fashion.
We have
\begin{equation}\label{ccnr-bound}
\begin{split}
 \norm{R(\tau)}_1&\ge \norm{R(\rho)}_1-r\norm{R(X)}_1 \\
  &\ge \norm{R(\rho)}_1-r\sqrt{d}\,,
\end{split}\end{equation}
where we first applied the triangle inequality, and for the last inequality we used that
\begin{equation}\label{lemma2}
 \max\set{\norm{R(X)}_1 | X\in \mathcal B'} = \sqrt{d}
\end{equation} 
%
(we refer to Appendix~\ref{app:proofs} for a proof).

Our goal is to find a state $\rho$ such that, if we depart from it by a distance $r$ in any direction, the resulting $\tau$ still fulfils PPT and violates CCNR, that is, $\minev[\Gamma(\tau)]\geq 0$ and $\norm{R(\tau)}_1>1$.
Then, using Eqs.~\eqref{ppt-bound} and \eqref{ccnr-bound}, we search for a state $\rho$ which admits the largest $r$ under the constraints
\begin{subequations}
\begin{align}
 \minev[\Gamma(\rho)]&\ge r\sqrt{(d-1)/d}\,,\label{opt1}\\
 \norm{R(\rho)}_1&> 1+r\sqrt{d}\,.\label{opt2}
\end{align}
\end{subequations}
For any admissible $r$, all states in $\mathcal B$ are bound entangled. 
Note that, while the optimization is naturally performed over physical target states $\rho$, the resulting ball $\mathcal B$ can well be partly outside of the state space, cf. Fig.~\ref{fig:fig1}, as the positivity of all states 
inside $\mathcal{B}$ is not imposed as a constraint.

\begin{figure}[ht]
	\centering
	\includegraphics[width=0.8\columnwidth]{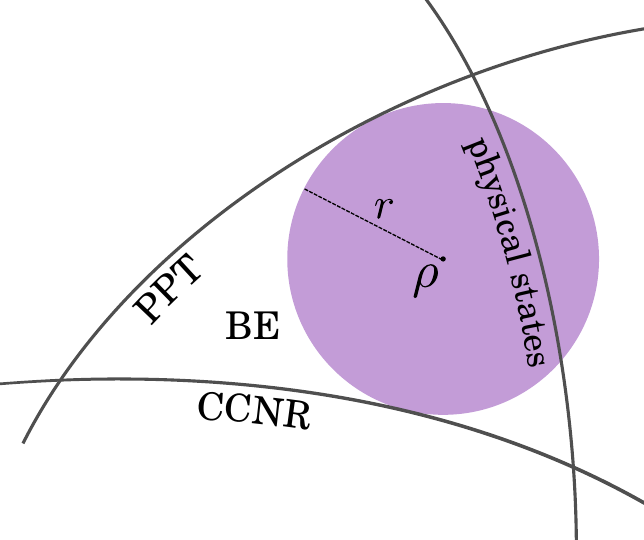}
	\caption{Schematic picture of the state space, where the boundaries of the PPT, CCNR and physical sets enclose a bound entangled region (BE). 
A region of points $\tau=\rho+r X$, where $X$ is a traceless Hermitian matrix with $\norm{X}_2\le 1$, is depicted in violet. The parameter $r$ is the radius of a Hilbert--Schmidt ball around a target state $\rho$ such that all physical states $\tau$ are bound entangled.
	}\label{fig:fig1}
\end{figure}

\section{Symmetric families}\label{sec:families}

The method described in the previous section is completely general, but for the optimization to actually become feasible one has to restrict the free parameters in $\rho$. We consider two symmetric families of bipartite states which are characterized by few parameters, and nevertheless contain fairly large regions of bound entanglement.

The first case that we consider is a family of two-qutrit states of the form
\begin{multline}\label{qutrits}
 \rho= a\ketbrad{\phi_3}+ b\sum_{k=0}^2{\ketbrad{k,k\oplus1}}\\
       +c\sum_{k=0}^2\ketbrad{k,k\oplus2} \,,
\end{multline}
where $\ket{\phi_3}= \sum_{k=0}^2 \ket{kk}/\sqrt3$, the symbol $\oplus$ denotes addition modulo 3, and $a,b,c$ are real parameters. 
A similar 
three-parameter family of two qutrits, known to contain bound entanglement and arise from symmetry conditions, was analyzed in Refs.~\cite{Baumgartner2006,Bertlmann2007, Bertlmann2008,Bertlmann2008a,Bertlmann2009,Sentis2016b}, and considered in the experiment reported in Ref.~\cite{Hiesmayr2013a}.
We search for the optimal values of $\{a,b,c\}$ that parametrize an optimal target state $\rho^*$. This state will admit any $r\leq r^*$, where $r^*$ is the maximum radius compatible with the constraints \eqref{opt1} and \eqref{opt2}. An exact solution for this optimization can be found (see Appendix~\ref{app:qutrits} for details). We obtain the optimal parameters 
\begin{equation}\label{qutrits_opt}
\hspace{-.5em}a\approx 0.21289\,,\; b\approx 0.04834\,,\text{ and}\; c\approx 0.21403 \,,
\end{equation}
yielding a maximal radius $r^*\approx 0.02345$. The resulting $\rho^*$ is a rank-7 state.

Since the ball $\mathcal{B}$ may contain unphysical states, it is in principle possible that our estimate for the maximal $r^*$ is not tight.
To explore this, we move away from $\rho^*$ by a distance $r^*$ in suitable directions and evaluate the position of the resulting state with respect to the boundaries of the PPT and the CCNR sets. We provide the details of this analysis in Appendix~\ref{app:qutrits}. As a result, we obtain that our estimate of $r^*$ is indeed tight with respect to the PPT boundary, but we observe that it slightly underestimates the distance with respect to the CCNR boundary.

As our second case, we consider the two-ququart states that are Bloch-diagonal, i.e., of the form
\begin{equation}\label{blochdiag-ququarts}
\rho = \sum_k x_k g_k \otimes g_k\,, 
\end{equation}
where 
$g_k = (\sigma_\mu\otimes\sigma_\nu)/2$, 
$\mu,\nu=0,1,2,3$, with $\sigma_0=\openone$, $\sigma_1,\sigma_2,\sigma_3$ the Pauli matrices, and 
the index $k$ enumerates pairs of indices $\{\mu,\nu\}$ in lexicographic order. Bound entangled Bloch-diagonal states have already been described in Ref.~\cite{Moroder2012}. The optimization over this family of states, despite having more free parameters, is much simpler than in the two-qutrit case discussed above. The reason is that the problem can be reshaped as a linear program over the coefficients $x_k$, 
and thus it can also be solved exactly (see Appendix~\ref{app:ququarts} for details). A vertex enumeration of the corresponding feasibility polytope is possible, and leads us to a set of 4224 optimal states achieving a maximal radius $r^*\approx 0.0214$. One example of such optimal state, $\rho^*$, is given by coefficients 
\begin{equation}\label{ququarttarget}
\begin{split}
x_1 &= \frac{1}{4} \,,\\
x_\alpha &\approx -0.0557066 \,,\quad \alpha\in \{2, 3, 4, 5, 6, 9, 12, 14, 16\} \,,\\
x_\beta &\approx 0.0142664 \,,\quad \beta\in \{7, 11, 13\}\,,\\
x_\gamma &\approx 0.0971467 \,,\quad \gamma\in \{8,10,15\} \,.
\end{split}
\end{equation}
where $x_1$ is fixed by normalization.
The state $\rho^*$ has rank 10, which is the minimal rank among bound entangled states achieving $r^*$ that are of the form~\eqref{blochdiag-ququarts} and are detectable by CCNR. As a byproduct of our analysis, we also obtain that the overall minimal rank for such bound entangled states is 9, albeit with a fairly smaller radius.

\section{Statistical analysis}
In this section we put our method for finding optimal target states to work in a practical scenario. That is, for an experiment aiming at the certification of bound entanglement, we design a statistical analysis of the experimental data that crucially hinges on knowing the radius of the bound entangled ball around the target state. The idea is to judge whether the data was obtained from a bound entangled state by considering the membership of the preparation to the bound entangled ball. In order to endow the certification with a measure of statistical significance, we design a hypothesis test for this membership problem.

Our null hypothesis, $H_0$, is that 
the prepared state is outside the bound entangled ball $\mathcal{B}(\rho_0,r_0)$ of radius $r_0$ centered at the target state $\rho_0$.
We make the assumption that an instance of experimental data, $\vec{x}$, is drawn from a multivariate normal distribution $N(\vec\xi,\Sigma)$ with offset $\vec\xi$ and covariance matrix $\Sigma$. This is a good approximation for realistic scenarios.
Here the vector notation is used over variables that belong to the same space as the experimental data, that is, e.g., the space of frequencies of measurement outcomes.
When $H_0$ holds true, then the offset $\vec\xi$ is the expected value of the 
data when a state $\rho_{\rm exp}$ is prepared such that 
$\norm{\rho_0-\rho_{\rm exp}}_2\geq r_0$.
The covariance matrix $\Sigma$ is determined by the particular experimental 
procedure used.

The goal is to design an appropriate hypothesis test for $H_1$ of the form 
$\hat{t}(\vec x)\le t$, where $t$ is a threshold parameter, and the function 
$\hat{t}$ is called a test statistic.
If the hypothesis test is true, the data $\vec x$ is accepted as fulfilling the 
hypothesis $H_1$, that is, the state is bound entangled.
The quantity through which we assess the significance of the hypothesis test is 
the worst-case probability of the test accepting $H_1$ while $H_0$ is true.
This is formally written as
\begin{equation}\label{p-value-def}
p(t,r_0) = \sup_{\rho}\left\{ \mathbb{P}[\hat{t}\le t] \mid \norm{\rho_0-\rho}_2 
\geq r_0\right\},
\end{equation}
where $\mathbb{P}[\hat{t}\le t]$ is the probability for the hypothesis test to 
accept data sampled from $N(\vec \xi,\Sigma)$,
$\vec\xi$ depends to $\rho$, and $\norm{\rho_0-\rho}_2\ge r_0$ is the assertion 
that $H_0$ holds true.
For given data $\vec x$, the probability $p[\hat{t}(\vec x), r_0]$ is the
 $p$-value of this hypothesis test.
 
Now, let us define %
\begin{equation}\label{test-statistic}
\hat{t}\colon \vec x\mapsto\norm{\Sigma^{-1/2}[T(\rho_0)-\vec x]}_2 \,,
\end{equation}
where $T$ is a map that takes a density matrix to the expected data (e.g., to a vector of probabilities),
thus it is determined by the experimental procedure. Hence, $\hat{t}(\vec x)$ 
gives some notion of distance between the standardized versions of the 
experimental data and the expected data of a perfect measurement performed over 
the target state. In Appendix~\ref{app:hyptest} we show that, if $\hat{t}$ is of 
the form in Eq.~\eqref{test-statistic}, the probability~\eqref{p-value-def} is 
naturally upper-bounded by
\begin{equation}\label{CDF}
p(t,r_0) \leq q_m(t^2,r_1^2) \,,
\end{equation}
where $s\mapsto q_m(s,u)$ is the cumulative distribution function of the noncentral $\chi^2$-distribution with $m$ degrees of freedom and noncentrality parameter $u$, 
and $r_1$ is the equivalent distance in the experimental data space of the Hilbert--Schmidt distance $r_0$, or, more precisely,
$r_1 := \inf_{\Delta\not\in\mathcal B} \hat{t}[T(\Delta)]$, where $\Delta$ is any Hermitian matrix with unit trace.
Naturally, the value of $r_1$ scales with $r_0$ and strongly depends on the experimental procedure, that is, on $T$ and $\Sigma$.

In the following, we show how to evaluate the hypothesis test for the two-qutrit state from Eq.~\eqref{qutrits_opt} as an example of a target state $\rho_0$, and assess the necessary experimental requirements for a desired level of significance. For this, we need to make some assumptions about the experimental procedure. 
We associate the measurement outcomes in the experiment to semidefinite-positive Hermitian operators $E_k$. 
We consider a complete set of such operators, 
that is, the real linear span of $\{E_k\}$ is the set of all Hermitian operators. 
The probability of obtaining the outcome corresponding to $E_k$ when measuring the state $\rho$ is given by $p_k=\tr(E_k\rho)$, and we assume that $T(\rho)_k\equiv p_k$. 
The connection between the probabilities $p_k$ and the data gathered in the experiment crucially depends on how the experiment is performed. A straightforward theoretical association can be established if one assumes that measurement outcomes correspond to independent Poissonian trials with parameters $nT(\rho)_k$, renormalized by $n$, where $n$ is the mean total number of events per measurement setting. 
That is, if we obtain $n_k$ events for the outcome $k$, we use as data $x_k= n_k/n$. Note, however, that in a realistic experiment $T(\rho)$ and $n$ are not known exactly.
A reasonable alternative is to use the total number of observed events $\sum_k n_k$ as an estimate of $n$, but then 
the connection between $p_k$ and $x_k$ is more involved. For our examples below we use this latter approach (a detailed discussion is presented in Appendix~\ref{app:hyptest}).

In order to get specific predictions, we assume that mutually unbiased
bases 
are used as local measurements (refer to Appendix~\ref{app:hyptest} for an explicit construction) and that the experimental state $\rho_{\rm exp}$ has 5\% white noise over the target, that is, $\rho_{\rm exp}= 0.95 \rho_0+0.05
\openone/9$. 
This amount of noise still results in a state within $\mathcal{B}$, since $\norm{\rho_0-\rho_{\rm exp}}_2\approx 0.6r_0$,  
where $r_0\approx 0.02345$ is the optimal radius associated to $\rho_0$ (cf. Section~\ref{sec:families}).
Then, one obtains $r_1^2\approx 0.0664^2n$ (see Appendix~\ref{app:hyptest} for details on how to compute this value).  
We now set a critical $t_0$ below which the $p$-value will be larger than some acceptable threshold $p_0$, so that $q_m(t_0^2,r_1^2)=p_0$. 
To obtain $t_0$, one inverts the equation $q_m(t^2,r_1^2)=p$ to get $t^2(p)$, called the quantile function. Then, $t_0 := t(p_0)$.
Once $t_0$ is fixed,
we compute the probability $p_{\rm fail}$ that the test fails to certify bound entanglement over data obtained by measuring the prepared state $\rho_{\rm exp}$, i.e., that $\hat{t}(\vec\xi)>t_0$ given $\vec\xi=T(\rho_{\rm exp})$. We can write this probability as 
\begin{equation}
p_\mathrm{fail}= 1-q_m\!\left(t_0^2, r_2^2\right) \,,
\end{equation}
where 
$r_2^2:={\hat{t}(\vec\xi)}^2 \approx 0.0416^2n$. 
Note that, while $r_1$ is used to determine $t_0$ considering a worst case scenario for a false positive, $r_2$ is a distance between standardized probabilities given the particular preparation $\rho_{\rm exp}$. 
Therefore, $r_1>r_2$ means that 
the test $\hat{t}$ has a chance to single out the particular experimental preparation as bound entangled from the worst-case state outside $\mathcal B$,
and hence $p_{\rm fail}$ will decrease with the number of samples $n$, which is the case of interest.
In Fig.~\ref{fig:fig2} we plot $p_{\rm fail}$ as a function of $n$, for various levels of significance $p_0$ expressed as multiples of the standard deviation $k\sigma$, so that $p_0=1-{\rm erf}(k/\sqrt{2})$. 
\begin{figure*}[ht]
	\centering
	\subfloat[]{{\includegraphics[width=.45\textwidth]{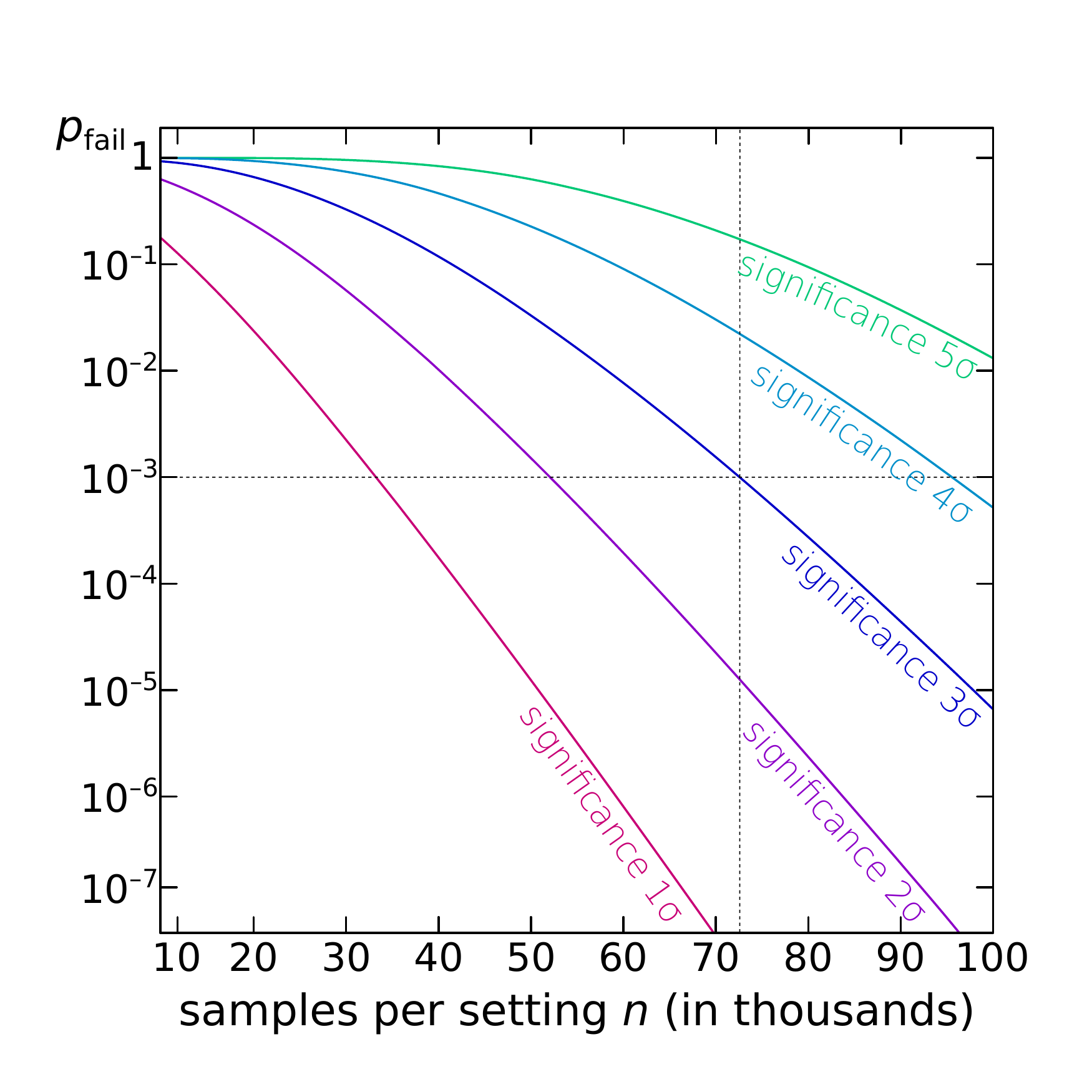}}}
	\hfill
	\subfloat[]{{\includegraphics[width=.45\textwidth]{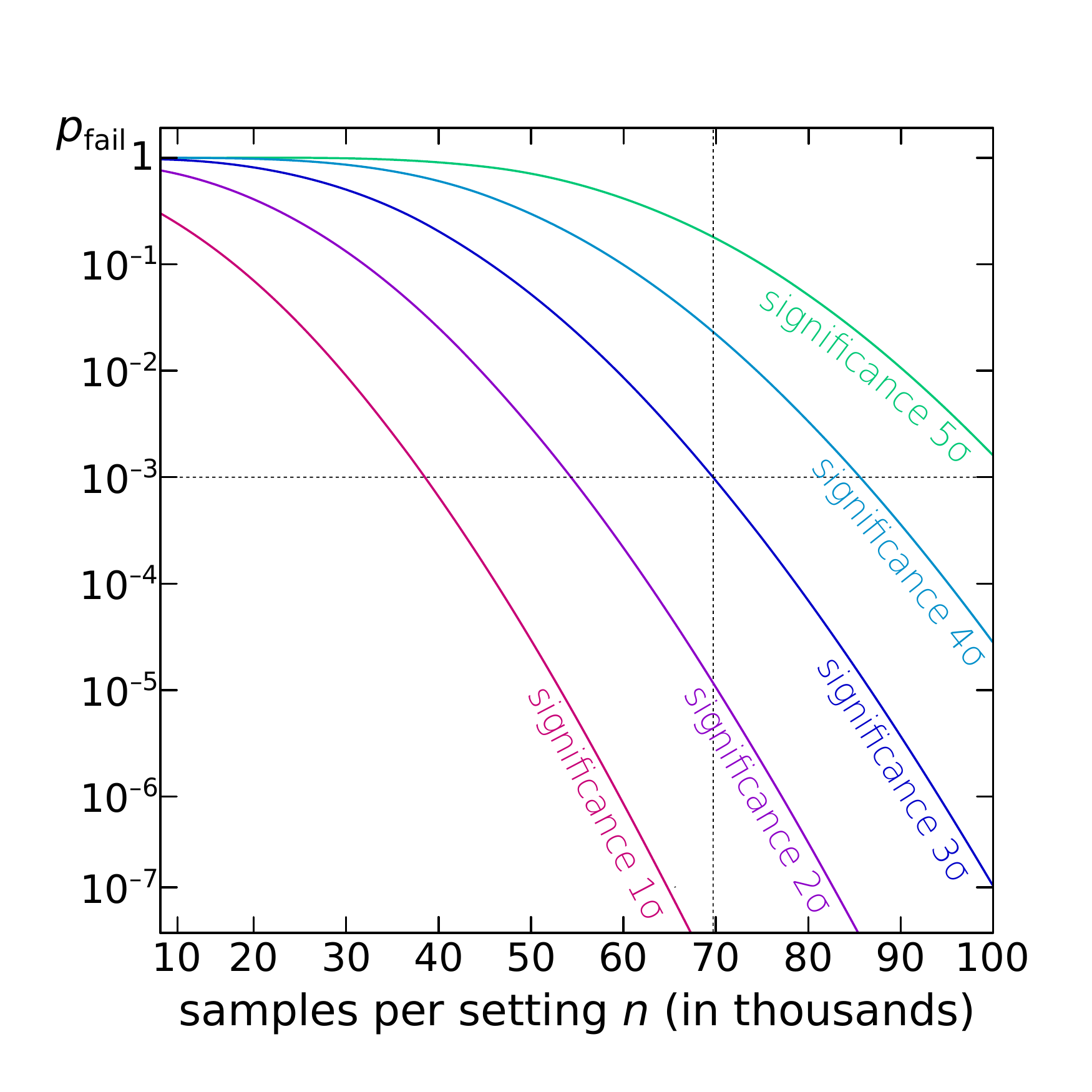}}}
\caption{\label{fig:fig2}
Probability $p_\mathrm{fail}$ to obtain data which does not confirm
bound 
entanglement at a level of significance corresponding to at least
$k\sigma$ 
standard deviations. The experiments consist of locally measuring
(a) mutually unbiased bases over preparations of the two-qutrit state in Eq.~\eqref{qutrits_opt} with 5\% white noise, and
(b) local Pauli measurements over the two-ququart state in Eq.~\eqref{ququarttarget} with 2.5\% white noise.}
\end{figure*}

Similarly, we carry out the same analysis for the two-ququart target state specified in Eq.~\eqref{ququarttarget}. In this case we construct our measurement settings 
by regarding each ququart as two qubits and performing local Pauli measurements on them. With $r_1^2\approx 0.0856^2 n$, $r_2^2\approx 0.0469^2 n$, and an admixture of 2.5\% white noise over the target $\rho_0$, we obtain the results shown in Fig.~\ref{fig:fig2}. 
In both cases, for the selected target states of two qutrits and two ququarts given by Eqs.~\eqref{qutrits_opt} and \eqref{ququarttarget}, our analysis shows that their experimental certification as bound entangled states under realistic assumptions is within reach, with around 70000 samples per measurement setting to reach a 3$\sigma$ level of significance.
Note that this number of samples justifies a posteriori our analysis assuming that experimental data $\vec x$ is normal-distributed.

\section{Discussion}

We have developed a comprehensive method for the experimental characterization of bipartite bound entanglement, from the selection of robust target states to the statistical analysis of the data. Previous experimental works were based on preparing a bound entangled state and inferring its properties from those of the reconstructed density matrix via maximum likelihood or least squares. Unfortunately, such techniques have been shown to produce unreliable results, particularly for entanglement certification~\cite{Sugiyama2012,Schwemmer2015}, which casts a shadow of doubt over past experimental demonstrations of bound entanglement. Instead of using a (necessarily) biased reconstruction of the density matrix and assuming that it shares the same properties that the true prepared state, we show that it is feasible to perform a hypothesis test over the unprocessed experimental tomographic data to test for membership of the preparation to a subset of the state space that is guaranteed to only contain bound entangled states. 
The hypothesis test is naturally equipped with a measure of statistical significance, which is easy to compute. 
The subset of bound entangled states is specified as a ball of radius $r$, 
and the design of the hypothesis test directly depends on this parameter. We 
have shown, through explicit examples of families of bipartite qutrit states 
and bipartite ququart states, how the certification of bound entanglement with 
high statistical significance is well within current experimental capabilities 
by using our method.

While our statistical analysis of the data in the form of a hypothesis test is a standalone technique applicable to the preparation and measurement of any state, we would like to stress that the value of the radius $r$ has an important effect in its detection power, hence it is worth aiming at target states with maximal $r$. To show this, we run our analysis for the two-qutrit state prepared in Ref.~\cite{Hiesmayr2013a}, assuming a noiseless experimental preparation. This state has a ball of bound entangled states around it of radius $r\approx 0.01182$. This value is computed as the maximal $r$ that satisfies Eqs.~\eqref{opt1} and \eqref{opt2}. We obtain that the required number of samples for achieving a $3\sigma$ level of significance with a failure probability $p_{\rm fail}\approx 10^{-3}$ is roughly twice as much as for the two-qutrit state in Eq.~\eqref{qutrits_opt} (which admits a radius $r\approx 0.02345$) when assuming a noisy preparation. As a comparison, optimizing over the one-parameter Horodecki family~\cite{Horodecki1998}---which is part of the family in Eq.~\eqref{qutrits}---yields a radius $r\approx 0.01681$, for the parameters $a\approx 0.28571$, $b\approx 0.07931$, and $c\approx 0.15879$. 

As a concluding remark, some comments on the generality of our result are in order. If one has prior knowledge of the expected amount and type of experimental noise, this could be incorporated into the statistical analysis by assuming the noisy state as the target. As a result, one should expect a smaller value of the test statistic $\hat{t}(\vec x)$ for a given set of data $\vec x$ and hence a smaller $p$-value. However, our calculations for several examples indicate that this advantage does not compensate in general the drawback of having a smaller bound entangled ball. A refinement of our method could also be achieved by taking into account the form of the covariance matrix $\Sigma$ into the optimization over target states, as we did incorporate it into the design of the hypothesis test. This would generally yield an ellipsoid around the target, instead of a ball, potentially capturing a larger volume of bound entangled states and thus leading to a stronger test. The possible disadvantage is that the bounds \eqref{opt1} and \eqref{opt2} will likely be much more complicated (if computable at all), hence the optimization step will be much harder to carry out. We leave the question open of whether such refinement provides a significant reduction in the experimental requirements. 

\vspace*{1cm}

\acknowledgments 
We acknowledge Otfried G\"uhne, G\'eza T\'oth, Beatrix Hiesmayr, Wolfgang L\"offler, and Dagmar Bru{\ss} for useful discussions.
This project was funded by the ERC Starting Grant No. 258647/GEDENTQOPT, the Basque Government grant No. IT986-16, the Spanish MINECO/FEDER/UE grant FIS2015-67161-P, the UPV/EHU program UFI 11/55, the ERC Consolidator Grant 683107/TempoQ, the DFG (No. KL 2726/2-1), and the Beijing Institute of Technology Research Fund Program for Young Scholars.

\onecolumn
\newpage
\appendix

\section{Proofs of Eqs.~\eqref{lemma1} and \eqref{lemma2}}\label{app:proofs}

Given a target state $\rho$, we construct the displaced state $\tau=\rho+r X$, where $r\geq 0$, and $X$ is a bounded traceless Hermitian matrix fulfilling $\norm{X}_2\leq 1$. We claim that the minimal eigenvalue of $\tau$ can be lower bounded as [cf. Eq.~\eqref{first_lower_bound}]
\begin{equation}
\minev(\tau)\ge \minev(\rho)-r\norm{X}_\infty \ge \minev(\rho)-r\sqrt{(d-1)/d}\,,
\end{equation}
where for the second inequality we have used that [cf. Eq.~\eqref{lemma1}]
\begin{equation}
\max\set{\norm{X}_\infty| X\in \mathcal B'}=  \sqrt{(d-1)/d} \,,
\end{equation}
and $\mathcal B'=\{X|\norm{X}_2\le 1,\tr X=0,X=X^\dagger\}$.
This holds from the following reasoning.
Clearly, the maximum is attained for $X=\diag(x,y_1,\dotsc,y_{d-1})$ with some vector $\vec{y}$ fulfilling 
 $\norm{\vec y}_\infty\le x= -\vec y\cdot \vec e$ and $x^2+\vec y^2\le 1$.
Here, $e_k=1$ for $k=1,\dotsc,d-1$. Note that the choice of $\vec y$ that allows $x$ to be largest is the uniform vector $\vec y = -x\vec e/\vec e^2 \equiv\vec z$, since it minimizes ${\vec y}^2$.
The maximization now reduces to 
$$\max\set {x| x^2+x^2/\vec e^2\le 1}\,,$$ 
which immediately yields the assertion due to $\vec e^2=d-1$.
 
In a similar fashion, we then obtain a lower bound of the 1-norm of the realigned state $\tau$ as [cf. Eq.~\eqref{ccnr-bound}]
\begin{equation}
\norm{R(\tau)}_1 \ge \norm{R(\rho)}_1-r\norm{R(X)}_1 \ge \norm{R(\rho)}_1-r\sqrt{d}\,,
\end{equation}
where we use that [cf. Eq.~\eqref{lemma2}]
\begin{equation}
\max\set{\norm{R(X)}_1 | X\in \mathcal B'} = \sqrt{d}\,.
\end{equation}
This immediately follows from the facts that
the $1$-norm is bounded by the $2$-norm as $\norm{A}_1\le 
\sqrt{\rank(A)}\norm{A}_2$, 
via e.g. Cauchy-Schwarz inequality,
and that $\norm{R(X)}_2=\norm{X}_2\leq 1$.
Hence $\sqrt d$ is an upper bound.
It remains to show that this bound is attainable.
For arbitrary $d$, we see that $\tr X = 0 \Leftrightarrow [R(X)]_{1,1} = 0$ (choosing the local basis such that $g_1 = \openone$). Then, $R(X)$ is 
a constant anti-diagonal matrix.

\section{Optimally detectable bound entangled states}

\subsection{A state of two-qutrits}\label{app:qutrits}
The first case that we consider is a family of two-qutrit states of the form
\begin{eqnarray}
\rho=a\ketbrad{\phi_3}&+&b\sum_{k=0}^2\ketbrad{k,k\oplus1}\nonumber\\
&+&c\sum_{k=0}^2\ketbrad{k,k\oplus2}\,,
\end{eqnarray}
where $\ket{\phi_3}=\sum_{k=0}^2\ket{kk}/\sqrt{3}$, the symbol $\oplus$ denotes addition modulo 3, and $a,b,c$ are real and nonnegative. We denote the dimension of the total Hilbert space by $d=9$. There are three distinct eigenvalues for $\rho$,
\begin{equation}
\lambda(\rho) = \{a,b,c\}\,,
\end{equation}
which satisfy the constraint
\begin{equation}\label{app:q3norm}
a+3(b+c)=1\,.
\end{equation}
Then, the three distinct eigenvalues of its partial transpose $\Gamma(\rho)$ are given by
\begin{equation}\label{app:q3_eigT}
\lambda(\Gamma(\rho))=\left\{\frac{a}{3},\frac1{6}\Bigl(1-a\pm\sqrt{4a^2+9(b-c)^2}\Bigr)\right\}\,.
\end{equation}
Now, the trace-norm of the realigned matrix $R(\rho)$ is given by
\begin{equation}\label{app:q3_CCNR}
\norm{R(\rho)}_1=\frac1{3}+2a+2\sqrt{3(b^2+c^2+bc)-(b+c)+\frac1{9}}\,.
\end{equation}
Plugging in Eqs.~\eqref{app:q3_eigT} and \eqref{app:q3_CCNR} in the bounds \eqref{opt1} and \eqref{opt2} one obtains two inequalities that, together with Eqs.~\eqref{app:q3norm} and $\lambda(\rho)\ge 0$, form the set of constraints that a valid target state should obey. The maximization over $r$ such that these constraints hold can be solved exactly, although the analytical form of the optimal parameters is rather involved. Here we give the approximate values $a\approx0.21289$, $b\approx0.04834$, and $c\approx0.21403$ that yield an optimal state $\rho^{\ast}$, which admits any $r\le r^* \approx 0.02345$. Note that the optimization is invariant under the interchange $b\leftrightarrow c$, therefore $r^*$ is also not affected by it. A visual representation of the density matrix of $\rho^*$ is shown in Fig.~\ref{fig:qutritsBarPlot}.

\begin{figure}[t]
\centering
	\includegraphics[scale=.75]{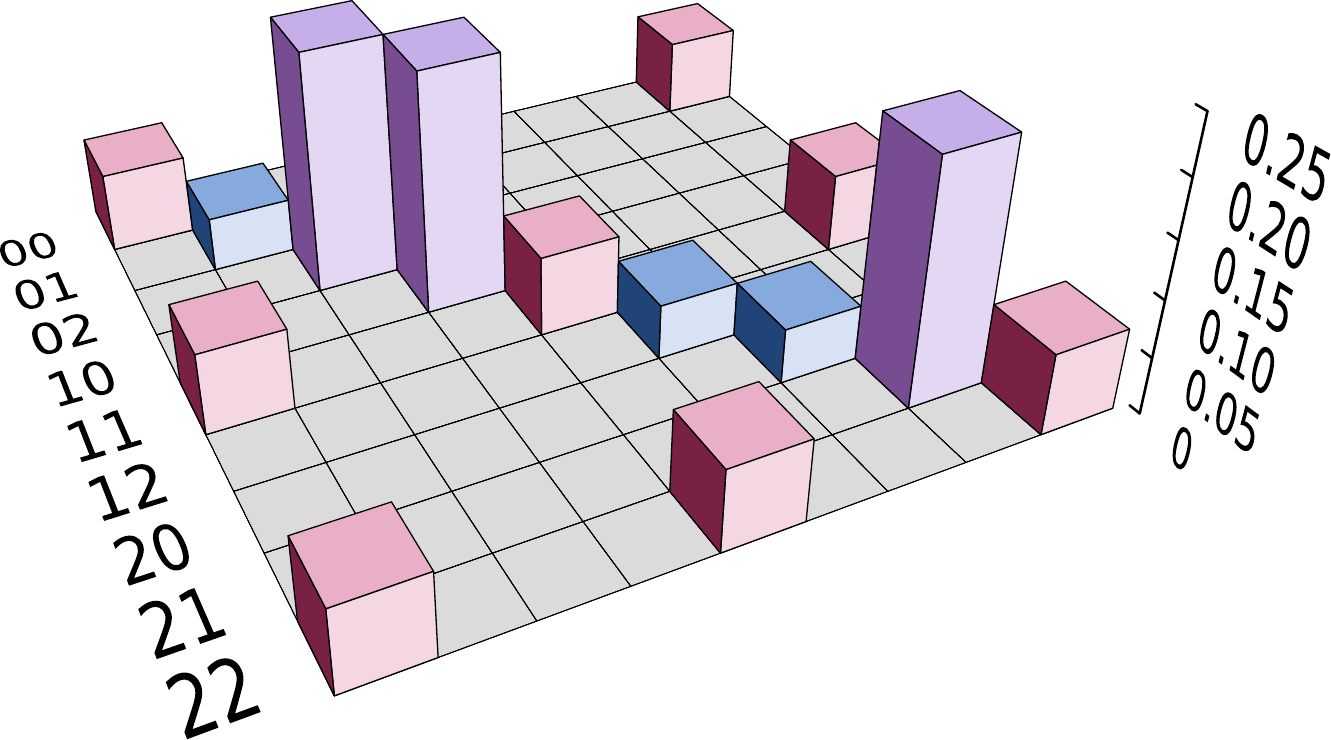}
\caption{The density matrix of our optimal two-qutrit state $\rho^*$ in the computational basis. 
The heights of the bars correspond to the entries of the density matrix. Note that all entries of $\rho^*$ are real-valued.
}\label{fig:qutritsBarPlot}
\end{figure}

As argued in the main text, since we do not require in our optimization that the full ball $\mathcal{B}$ is contained in the state space, an analysis of the tightness of our estimate for $r^*$ is in order. To do so,
we move away from $\rho^*$ by a distance $r^*$ in suitable directions and evaluate the position of the resulting state with respect to the boundaries of the PPT and the CCNR sets. We first move towards the PPT boundary. To this end, we take a normalized eigenvector $\ket\eta$ from the eigenspace of $\Gamma(\rho^*)$ with minimal eigenvalue and find, for the matrix after the displacement $\tilde\rho= 
 \rho^*-r^*\Pi[\Gamma^{-1}(\ketbrad\eta)]$, the values
\begin{equation}
 \minev[\Gamma(\tilde\rho)]\approx 0\,,\quad
 \norm{R(\tilde\rho)}_1\approx 1.094\,,\quad \text{ and}
 \quad \minev(\tilde\rho)\approx -0.005\,,
\end{equation}
where we wrote $\Pi[X]$ for $[X-\tr(X)\openone/9]/\xi$ for some appropriate $\xi$, so that $\norm{\Pi[X]}_2=1$.
Hence, $\tilde\rho$ is effectively sitting on the PPT boundary, it is still detected as bound entangled by the CCNR criterion, and it is slightly outside the space of physical states.
This shows that our estimate for $r^*$ is tight with respect to PPT.

Similarly, we now consider a displacement towards the CCNR boundary.
With a singular value decomposition $UDV^\dag=R(\rho^*)$, we let 
 $\bar\rho=\rho^*-r^*\Pi[R^{-1}(UV^\dag)]$ and find the values
\begin{equation}
 \minev[\Gamma(\bar\rho)]\approx 0.03\,,\quad
 \norm{R(\bar\rho)}_1\approx 1.004\,,\quad\text{ and}
 \quad \minev(\bar\rho)\approx 0.01\,.
\end{equation}
Therefore, it is possible that our state $\rho^*$ is not yet optimal for obtaining the largest value $r^*$, since we are still not touching the CCNR boundary.

\subsection{States of two-ququarts}\label{app:ququarts}
We consider two-ququart states that are Bloch-diagonal, i.e., of the form
\begin{equation}\label{app:blochdiag-ququarts}
\rho=\sum_k x_k g_k\otimes g_k\,,
\end{equation}
where $g_k = (\sigma_\mu\otimes\sigma_\nu)/2$, $\mu,\nu=0,1,2,3$, with $\sigma_0=\openone$, $\sigma_1,\sigma_2,\sigma_3$ the Pauli matrices, and 
the index $k$ enumerates pairs of indices $\{\mu,\nu\}$ in lexicographic order. Since $\tr(\rho)=1$, we get $x_1=1/4$. 

By making the assumption that $|x_k| = s$ for $k=2,\ldots,16$, the optimization can be easily carried out analytically. We will see later that the maximum radius $r^*$ for the general states \eqref{app:blochdiag-ququarts} is indeed achieved under this assumption. 
The trace norm of the realignment is
\begin{equation}
\norm{R(\rho)}_1=\sum_k |x_k|=\frac1{4}+15s\,.
\end{equation}
Computing the minimal eigenvalue of $\Gamma(\rho)$ is not so straightforward, as the signs of the coefficients in Eq.~\eqref{app:blochdiag-ququarts} ought to be taken into account. We consider all possible sign combinations, and we see that the minimal eigenvalue of $\Gamma(\rho)$ is always of the form $\lambda_{\scriptsize\mbox{min}}[\Gamma(\rho)]=\frac1{16}[1-4(2k-1) s]$, where $k=1,2,\ldots,8$. Further, we check that $k=2$ is the only combination with which the constraints $\rho\geq 0$, $\Gamma(\rho)\geq 0$, and $\norm{R(\rho)}_1>1$ are satisfied, so we have
\begin{equation}
\lambda_{\scriptsize\mbox{min}}[\Gamma(\rho)]=\frac1{16}-\frac3{4}s\,.
\end{equation}
Let us denote 
\begin{subequations}\label{eq:r_ququart}
\begin{align}
r_a & := \sqrt{d/(d-1)}\,\lambda_{\scriptsize\mbox{min}}[\Gamma(\rho)]
= \frac{4}{\sqrt{15}}\left(\frac1{16}-\frac3{4}s\right)\,,\nonumber\\
r_b & := (\norm{R(\rho)}_1-1)/\sqrt{d} = \frac{1}{4}\left(15 s-\frac{3}{4}\right)\,.\nonumber
\end{align}
\end{subequations}
Since $r_a$ and $r_b$ are affine functions of the parameter $s$, the solution to the equation $r_a=r_b$ will single out a unique $s^*$, and the radius $r^*$ associated to it will be maximal. These are
\begin{align}
s^* &= \frac1{12}\frac{4+3\sqrt{15}}{4+5\sqrt{15}}\approx0.05571\,,\nonumber\\
r^* &= \frac{1}{4}\left(15 s^*-\frac{3}{4}\right)\approx0.0214\,.
\end{align}
As argued before, the parameter $s^*$ is not enough to characterize a state of the form~\eqref{app:blochdiag-ququarts}, as sign combinations of the coefficients $x_k$ matter. An example of a state that achieves $r^*$ is given by
the parameters $x_1=1/4$, $x_{k}=s^*$ for $k\in S=\{2,3,4,5,6,7,9,10,12,14\}$, and $x_{k}=-s^*$ for $k\in S^c\setminus\{1\}$.

Let us now tackle the optimization over the family of states in Eq.~\eqref{app:blochdiag-ququarts} without any extra assumption, for which we will resort to an exact but computer-aided method. First, note that the search for the optimal state can be concisely expressed as 
\begin{equation}\label{linprog-1}
\begin{split}
	{\rm maximize}\quad& r \\
	{\rm subject\;to}\quad& Mx +\frac{1}{4} \geq 0 \\
	& MDx + \frac{1}{4} \geq \sqrt{15}r \\
	& \sum_k |x_k|-\frac{3}{4} \geq 4r
\end{split}
\end{equation}
where $x=\{x_k\}_{k=2}^{16}$ is a 15-dimensional real vector, $r$ is a real nonnegative number, $D$ is a diagonal matrix of signs such that the map $x\mapsto Dx$ effectively induces the map $\rho\mapsto \Gamma(\rho)$, $M$ is a $16\times 15$ matrix of signs such that $\tfrac{1}{4}Mx+\tfrac{1}{16}$ is a vector of eigenvalues of $\rho$,
and recall that $x_1 =1/4$ is fixed by normalization. 
Under these considerations it becomes apparent that the constraints in Eq.~\eqref{linprog-1} correspond to semidefinite-positiveness, PPT, and CCNR, respectively.

In order to solve Eq.~\eqref{linprog-1}, we linearize the CCNR constraint by removing the absolute values and considering all the $2^{15}$ possible sign combinations in the sum. 
Moreover, we incorporate the constraint $r\geq 0$ and regard $r$ as an extra free parameter, instead of as an objective function.
Then,
each sign combination chosen in the CCNR constraint results in a linear constraint that, along with positivity, PPT, and $r\ge 0$, will define a polytope of feasible parameters
and, since all constraints are linear in $r$, its maximal value will occur for one of the vertices.
Of course, each single linear program of this kind is in principle more restrictive than Eq.~\eqref{linprog-1}, but it is granted that the solution we are looking for will correspond to a vertex of the feasibility polytope of one of them.
We enumerate the vertices of all the $2^{15}$ polytopes using the polyhedral computation library \texttt{cddlib}\footnote{\texttt{cddlib} version 0.94h, \url{https://www.inf.ethz.ch/personal/fukudak/cdd_home/}}.
We then take the union of all vertices, eliminate redundant and nonextremal ones, and end up with a set of 4224 vertices with a maximal radius $r^*\approx 0.0214$, which coincides with the value analytically obtained above. The advantage of this method is that we obtain a much larger set of optimal states with varying rank. Having optimal states with smaller rank can significantly simplify their experimental preparation. In contrast, note that any state for which the assumption made in the analytical calculation above, $|x_k|=s$ for all $k=2,\ldots,16$, holds, is full-rank. The minimal rank for states of the form~\eqref{app:blochdiag-ququarts} achieving a maximal radius $r^*$ is 10.
An example of such a rank-10 optimal state is given by parameters
\begin{equation}\label{eq:opt_q4}
\begin{split}
x_1 &= \frac{1}{4} \,,\\
x_\alpha &\approx -0.0557066 \,,\quad \alpha\in \{2, 3, 4, 5, 6, 9, 12, 14, 16\} \,,\\
x_\beta &\approx 0.0142664 \,,\quad \beta\in \{7, 11, 13\}\,,\\
x_\gamma &\approx 0.0971467 \,,\quad \gamma\in \{8,10,15\} \,.
\end{split}
\end{equation}
\begin{figure}[t]
\centering
	\includegraphics[scale=.75]{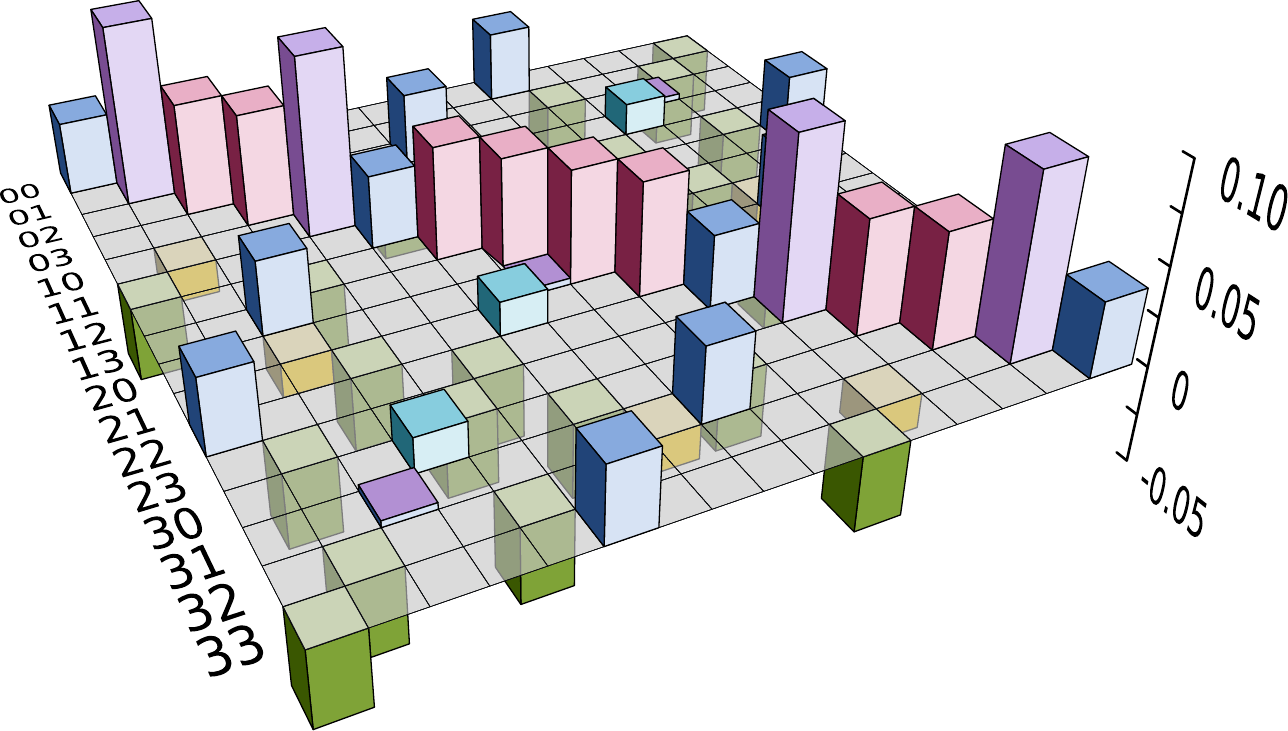}
\caption{The density matrix of the optimal two-ququart state with parameters given in Eq.~\eqref{eq:opt_q4} in the computational basis. The heights of the bars correspond to the entries of the density matrix. Note that these are all real-valued.
}\label{fig:ququartsBarPlot}
\end{figure}
In Fig.~\ref{fig:ququartsBarPlot} we depict the resulting density matrix as a bar plot.

The vertex enumeration of all feasibility polytopes contained in Eq.~\eqref{linprog-1} allows us to go beyond the set of optimal states and characterize states with even smaller ranks, naturally at the expense of sacrificing some volume of bound entangled states around them. 
With this in mind, we see that the absolute minimal rank for a Bloch-diagonal two-ququart bound entangled state is 9, with associated radius $r_9 \approx 0.0128$. An example of such state is given by parameters
\begin{equation}
\begin{split}
x_1 &= \frac{1}{4} \,,\\
x_\alpha &\approx 0.0999184 \,,\quad \alpha\in\{2, 5\}\,,\\
x_\beta &\approx 0.0750408 \,,\quad \beta\in\{3, 4, 6, 10, 14\}\,,\\
x_\gamma &\approx 0.0501632 \,,\quad \gamma\in\{7, 9\}\,,\\
x_\delta &\approx -0.0252856 \,,\quad \delta\in\{8, 13\}\,,\\
x_\epsilon &= -x_\gamma \,,\quad \epsilon\in\{11, 15, 16\}\,,\\
x_{12} &= -x_\delta \,.
\end{split}
\end{equation}

%
%
%

\section{A hypothesis test for bound entanglement}\label{app:hyptest}

In this appendix we give a proof of Eq.~\eqref{CDF} and show how to compute the parameters $r_1$, $r_2$, and $m$, needed to reproduce Fig.~\ref{fig:fig2}. Let us start with the proof.

The $\chi^2$-distribution arises naturally in a hypothesis test where the hypotheses are defined in terms of bounds on Euclidean distances in a vector space of normal-distributed random variables. 
To see this, we consider a generic situation where we draw a sample $\vec x$ from an $m$-variate normal distribution $N(\vec\xi, \Sigma)$ with offset $\vec\xi$ and covariance matrix $\Sigma$. The offset shall be of the form $\vec\xi=S\lambda+\vec\eta$ with a matrix $S$ and a constant vector $\vec\eta$. 
We denote by $p(t,r)$ the maximal probability that the hypothesis test $\hat{t}(\vec x)\le t$ holds true under the hypothesis $\norm{\bf \mu- \lambda}_2\ge r$ for some fixed $\mu$, where $\vec x$ is sampled from $N(\vec\xi,\Sigma)$, and $\hat t$ is a test statistic. Then, for a given sample $\vec x$, the probability $p[\hat{t}(\vec x), r]$ is a $p$-value for the hypothesis test $\hat{t}(\vec x)\le t$. By choosing
\begin{equation}\label{that}
 \hat{t}\colon \vec x \mapsto \norm{\Sigma^{-1/2}(S\mu+ \vec\eta- \vec x)}_2\,,
\end{equation}
the following lemma holds:

\begin{lemma}\label{getpvalue}
We have
\begin{equation}\label{pvalue}
 p(t,r)= q_m[t^2, r^2\, \lambda_\mathrm{min}(S^T\Sigma^{-1}S)]\,,
\end{equation}
where $s\mapsto q_m(s,u)$ is the cumulative distribution function of
the 
noncentral $\chi^2$-distribution with $m$ degrees of freedom and
noncentrality 
parameter $u$. Here, $\lambda_\mathrm{min}(X)$ denotes the smallest
eigenvalue 
of the symmetric matrix $X$.
\end{lemma}
\begin{proof}
By our assumptions, we have
\begin{equation}\label{yet-another-pvalue}
 p(t, r)= \sup_{\lambda}\set{ \mathbb{P}[\hat{t}\le t] | 
 \norm{\mu- \lambda}_2\ge r}\,,
\end{equation}
where $\hat t$ takes $\vec x$ distributed according to the normal distribution
$N(\vec\xi, \Sigma)$. Then, $\vec y= \Sigma^{-1/2}(\vec x- \vec\xi)$ is normal-distributed 
with offset $\vec 0$ and covariance matrix $\openone$, and we find
\begin{equation}\begin{split}
 p(t,r)&= \sup_{\upsilon}\set{
 \mathbb{P} [\norm{\vec y- \Sigma^{-1/2}S\upsilon}_2\le t : \vec y \sim N(\vec 0,\openone) ] | 
 \norm{\upsilon}_2\ge r} \\
 &\equiv \sup_{\upsilon}\set{
 q_m(t^2, \norm{\Sigma^{-1/2}S\upsilon}_2^2)
 | \norm{\upsilon}_2\ge r}\,,
\end{split}\end{equation}
where we used the definition of $q_m(s, u)$ in the second step. In
order to see 
that this yields Eq.~\eqref{pvalue}, it is enough to note 
that $q_m(s, u)$ is a monotonously decreasing function in $u$~\cite{Johnson1994}.
\end{proof}

The cumulative distribution function of the noncentral $\chi^2$-
distribution 
is given in terms of Marcum's $Q$-function as $q_m(s, u)= 1-
Q_{m/2}(\sqrt 
u,\sqrt s)$. In particular,
\begin{equation}\label{Isumform}
 q_m(s, u)= \mathrm e^{-u/2}\sum_{l=0}^\infty
\frac{(u/2)^l}{l!} 
 \tilde\gamma(l+\tfrac m2,\tfrac s2)\,,
\end{equation}
where $\tilde\gamma$ is the regularized lower incomplete gamma function, so that $s\mapsto q_m(s, 0)$ is the cumulative distribution
function of the central $\chi^2$-distribution\footnote{From Eq.~\eqref{Isumform}
it is 
also easy to see that $f_k\colon u\mapsto q_{2k}(s,2u)$ is a
decreasing 
function in $u$, since $\partial_u f_k(u)= f_{k+1}(u)-f_k(u)$ and 
$\tilde\gamma(a,z)- \tilde\gamma(a+1,z)= \partial_z\tilde\gamma(a+1,z)=
\mathrm 
e^{-z}z^a/\Gamma(a+1)\ge0$.}.

Lemma \ref{getpvalue} is presented in a generic form. 
At variance with the $p$-value defined in the main text [cf. Eq.~\eqref{p-value-def}], note that the supremum in Eq.~\eqref{yet-another-pvalue} is taken over a larger space, as $\lambda$ is not affected by any positivity constraint. Hence, Eq.~\eqref{pvalue} is an upper bound of Eq.~\eqref{p-value-def}. 
In the context of the quantum experiment at hand, 
$\mu$ and $\lambda$ are Hermitian operators, and 
we make the identifications $T(\rho_0) \equiv S(\mu) + \vec\eta$ and $T(\rho_{\rm exp}) \equiv S(\lambda) + \vec\eta$, where $\rho_0$ is the target state with associated radius $r_0$, and $\rho_{\rm exp}$ is the actual state prepared in the experiment. 
Then, we find that $r_0^2\, \lambda_\mathrm{min}(S^T\Sigma^{-1}S)\equiv r_1^2$ [cf. Eq. \eqref{CDF}].

We now assume that the data is of the form $x_k^\ell=n_k^\ell/\sum_j n_j^\ell$, 
where $n_k^\ell$ is the number of events for outcome $k$ and measurement 
setting $\ell$. Then, the elements of the vector $T(X)$ are $\tr(E_k^\ell X)$, 
given the effects $\{E_k^\ell\}$ with $\sum_k E_k^\ell=\openone$, i.e., $T$ 
takes quantum states to probabilities, and $S$ is the restriction of $T$ to the 
traceless operators. This determines $\vec\eta= T[\openone/\tr(\openone)]$. 
Note that one effect per measurement setting shall not be included in $S$, 
since its associated outcome will not be independent and this will lead to 
$r_1=0$.
Note also that introducing the restricted map $S$ is necessary in order to have a nonsingular covariance matrix $\Sigma$ with a well defined inverse.
With this in mind, one can determine the number of degrees of freedom $m$ for the examples of qutrits and ququarts described in the main text. In the case of qutrits, measuring locally a complete set of mutually unbiased bases\footnote{Two orthonormal bases $\{\ket{e_i}\}_{i=1}^d$ and $\{\ket{f_i}\}_{i=1}^d$ are called \emph{mutually unbiased} if it holds that $|\langle e_j | f_k\rangle|^2=1/d$ $\forall j,k=1,\ldots,d$.} (MUBs) in principle produces data of dimension 144 (we have 4 settings per party with 3 outcomes per setting), but eliminating dependent outcomes reduces this dimension to $m=128$. 
For the sake of completeness, we give an explicit construction of the MUBs $\{M_\ell\}_{\ell=0}^3$, where each row in $M_\ell$ corresponds to a basis vector. Then, $M_0$ is the identity matrix,
\begin{equation}
M_1=\frac{1}{\sqrt{3}}
\begin{pmatrix}
1 & 1 & 1 \\
1 & \omega^2 & \omega \\
1 & \omega & \omega^2
\end{pmatrix} \,,\quad
M_2=\frac{1}{\sqrt{3}}
\begin{pmatrix}
1 & 1 & \omega \\
1 & \omega^2 & \omega^2 \\
1 & \omega & 1
\end{pmatrix} \,,\quad
M_3=\frac{1}{\sqrt{3}}
\begin{pmatrix}
1 & 1 & \omega^2 \\
1 & \omega^2 & 1 \\
1 & \omega & \omega
\end{pmatrix} \,,
\end{equation}
and $\omega=\exp(2\pi i/3)$.
Likewise, for the ququart case we assume that the measurement settings for each party are built as pairs of local Pauli measurements $\{\sigma_\mu,\sigma_\nu\}$, with $\mu,\nu=1,2,3$, hence we have 9 settings per party with 4 outcomes per setting, which leads to a reduced dimension of $m=1215$.

For the moment, let us assume that the measurement outcomes $k$ for each 
setting $\ell$ follow a Poisson distribution with parameters $N^\ell 
T(\rho)^\ell_k$, so that data is now of the form $x^\ell_k=n^\ell_k/N^\ell$, 
where $N^\ell$ is the mean total value of events for measurement setting 
$\ell$. Then, the covariance matrix is a simple diagonal matrix with entries 
$T(\rho)^\ell_k/N^\ell$. In this situation, also the elimination of dependent 
outcomes is not necessary since $\Sigma$ is invertible.
This procedure is not optimal. If $N^\ell$ is large enough, we can well 
approximate $N^\ell$ by $\sum_k n_k^\ell$, at the cost of $\Sigma$ becoming 
singular. However, after the dimensional reduction introduced above, the 
reduced covariance matrix is again invertible. The advantage of this approach 
is that we reduced the dimension $m$ without discarding data, and hence the 
overall statistical performance is increased. This is due to the fact that 
$r_1$ and $r_2$ in this approach are the same as in the previous one. Note that 
an analysis using multinomial statistics yields the same reduced covariance 
matrix and hence the same parameters $m$, $r_1$, and $r_2$.
For our examples in the main text we use the latter approach and assume 
$N^\ell\equiv n$ for all settings $\ell$. The parameters used for computing 
$p_{\rm fail}$ in Fig.~\ref{fig:fig2} are computed to be $r_1^2\approx 
0.0664^2n$ and $r_2^2\approx 0.0416^2n$ for the qutrit example, and 
$r_1^2\approx 0.0856^2n$ and $r_2^2\approx 0.0469^2n$ in the ququart case. 

\newpage

\bibliographystyle{quantumrefs}
\bibliography{ExpBoundEntVerification}

\end{document}